\documentclass[reqno,12pt]{amsart}
\usepackage{amsmath,amssymb,amsthm,graphicx,mathrsfs,url,mathtools}
\usepackage[usenames,dvipsnames]{color}
\usepackage[colorlinks=true,linkcolor=Red,citecolor=Green]{hyperref}
\usepackage[letterpaper,hmargin=1in,vmargin=1in]{geometry}
\usepackage{amsxtra}
\usepackage{subcaption}
\usepackage{verbatim}
\usepackage{lipsum}
\usepackage{mwe}

\def\?[#1]{\textbf{[#1]}\marginpar{\Large{\textbf{??}}}} 
\def\smallsection#1{\smallskip\noindent\textbf{#1}.}
\let\epsilon=\varepsilon 

\newcommand{\CC}{{\mathbb C}}
\newcommand{\diag}{\mathrm{diag}}

\newcommand{\ZZ}{{\mathbb Z}}
\newcommand{\bt}{\mathbf{t}}
\newcommand{\Hn}{H_n ( \alpha)}

\newcommand{\kk}{k}

\newcommand{\bo}{0}
\newcommand{\bi}{i}
\newcommand{\bk}{k}

\newcommand{\magic}{\mathcal{A}}

\newtheorem{theo}{Theorem}
\newtheorem*{theo*}{Theorem}
\newtheorem{prop}{Proposition}[section]

\newtheorem{lemm}[prop]{Lemma}
\newtheorem{corr}[prop]{Corollary}
\newtheorem{rem}[prop]{Remark}
\numberwithin{equation}{section}

\DeclareMathOperator{\Spec}{Spec}

\DeclareMathOperator{\loc}{loc}

\DeclareMathOperator{\rank}{rank}

\DeclareMathOperator{\tr}{tr}
\DeclareMathOperator{\id}{\text{Id}}
\DeclareMathOperator{\coker}{coker}

\usepackage{scalerel}

\newcommand\reallywidehat[1]{\arraycolsep=0pt\relax%
\begin{array}{c}
\stretchto{
  \scaleto{
    \scalerel*[\widthof{\ensuremath{#1}}]{\kern-.5pt\bigwedge\kern-.5pt}
    {\rule[-\textheight/2]{1ex}{\textheight}} 
  }{\textheight} %
}{0.5ex}\\           
#1\\                 
\rule{-1ex}{0ex}
\end{array}
}
\title[Twisted multilayer graphene]{Flat bands, Dirac cones, and higher-order band crossings in twisted multilayer graphene}

\author{Bryan Li}
\email{byli@berkeley.edu}
\address{Department of Mathematics, University of California, Berkeley, CA 94720, USA.}

\author{Mengxuan Yang}
\email{yangmx@princeton.edu}
\address{Department of Operations Research \& Financial Engineering, Princeton University, Princeton, NJ 08544 USA}

\begin{document}

\begin{abstract}
For the chiral limit of two sheets of $n$-layer Bernal-stacked graphene established in the Physical Review Letters \cite{wang2022hierarchy,ledwith2022family}, we prove a trichotomy: depending on the twisting angle, we have either (1) generically, the band crossing of the first two bands is of order $n$; or (2) at a discrete set of magic angles, the first two bands are completely flat; or (3) for another discrete set of twisting parameters, the bands exhibit Dirac cones.  
This new mathematical discovery disproves the common belief in physics that such a twisted multilayer graphene model can only have higher order band crossings or flat bands, and it leads to a new type of topological phase transition. 
\end{abstract}

\maketitle

\section{Introduction}
We consider the chiral limit \cite{tarnopolsky2019origin} of two sheets of $n$-layer Bernal-stacked graphene twisted by a small angle \cite{wang2022hierarchy,ledwith2022family}. The twisted multilayer graphene (TMG) Hamiltonian is given by 
\begin{equation}
  \label{eq:defH}    
  H ( \alpha; \bt) \coloneq \begin{pmatrix} 0 & D_n(\alpha; \mathbf{t})^* \\
  D_n ( \alpha; \mathbf{t}) &  0 \end{pmatrix}    
\end{equation} 
where 
\begin{equation}
  \label{eq:defD}
  D_n ( \alpha; \bt ) \coloneq \left(\begin{matrix} D (\alpha) & t_1T_+ & \\ t_1T_- & D(0) & t_2T_+ & \\ & t_2T_- & D(0) & \ddots \\ && \ddots && \\ &&&& t_{n-1}T_+ \\ &&&t_{n-1}T_- & D(0)\end{matrix}\right),
\end{equation}
with $\bt=(t_1, t_2, \cdots, t_{n-1})$ and  
\begin{equation}
  \label{eq:defD1}
  D( \alpha ) \coloneq \begin{pmatrix} { 2 }  D_{\bar z}  &   \alpha U(z) \\  
    \alpha U(-z) & { 2 } D_{\bar z }   \end{pmatrix}, \ \ 
    T_+ = \left(\begin{matrix}1&0\\0&0\end{matrix}\right), \ \ 
    T_- = \left(\begin{matrix}0&0\\0&1\end{matrix}\right). 
\end{equation}
In particular, $ z = x_1 + i x_2$, $ D_{\bar z } := \tfrac1{2 i } ( \partial_{x_1} + 
i \partial_{x_2} )$ and 
\begin{equation}
\label{eq:defU} 
U (z) = U (z,\bar z ) \coloneq \sum_{ j = 0}^2 \omega^j e^{ \frac12 (  z  
\bar \omega^j - \bar z  \omega^j )} , \ \
\omega \coloneq e^{ 2 \pi i /3 }. 
\end{equation}
In equation \eqref{eq:defD}, $t_kT_+$ (resp.\ $t_kT_-$) denotes the tunneling between the top (resp.\ the bottom) $k$-th layer and $(k+1)$-th layer. Without loss of generality we assume that all $t_i$'s are non-zero, since otherwise we have direct sum decomposition of $D_n (\alpha;\bt)$ into smaller matrix blocks. In this paper, we consider the band structure when $\bt$ is fixed (see \cite{yang2023flat} for band structures of $H(\alpha;\bt)$ when $\bt$ varies).  Hence, we shall drop $\bt$ in the Hamiltonian and simply write $D_n(\alpha)$ or $H(\alpha)$ for notational convenience when there is no confusion. 

\begin{rem}
For $n=1$, the operator \eqref{eq:defH} reduces to the famous chiral limit of twisted bilayer graphene (TBG) model
\begin{equation}
\label{eq:H2}
H(\alpha) = 
\begin{pmatrix} 0 & D(\alpha)^* \\
  D ( \alpha) &  0 
  \end{pmatrix},
\end{equation}
see \cite{tarnopolsky2019origin,becker2020mathematics,becker2022fine} for more details. The dimensionless parameter $\alpha$ is essentially the reciprocal of the twisting angle in physics.
\end{rem}

Note that the potential \eqref{eq:defU} satisfies the following properties:
\begin{equation}
\label{eq:symmU} 
\begin{gathered}  
U ( z + \mathbf a ) = \bar \omega^{a_1 + a_2}  U ( z ) , \ \ 
\mathbf a = \tfrac{4}3 \pi i \omega (  a_1 + \omega a_2 ) ,  \ \ a_j \in \mathbb Z , 
\\ U ( \omega z )  = \omega U ( z )   , \ \  \overline { U ( \bar z ) } = U ( z ) . 
\end{gathered}
\end{equation}
In particular, it is periodic with respect to the lattice $3\Lambda$ with \begin{equation}
\Lambda := \left\{ \tfrac{4}{3} \pi i \omega( a_1 + \omega a_2 ) : a_1,a_2 \in \mathbb Z  \right\}.
\end{equation}
The dual lattice $\Lambda^*$ of $\Lambda$ consists of $ \kk \in\CC $ satisfying
\begin{gather} \langle \gamma,\kk \rangle := \tfrac12 (\gamma \bar{\kk } + \bar \gamma \kk ) \in 2 \pi \ZZ , \ \  \gamma\in \Lambda \Longrightarrow
 \Lambda^* = {\sqrt{3}} \omega \left(
  \ZZ  \oplus  {\omega}  \ZZ \right).
\end{gather}
We consider a generalized Floquet condition under the modified translation operator $\mathscr L_{\mathbf{a}}$, with $\mathbf{a}\in \Lambda$ (cf.~\eqref{eq:translation}), i.e., we study the spectrum of $ \Hn $ satisfying the following boundary conditions:
\begin{equation}
\label{eq:FL_ev}  
H(\alpha) \mathbf u = E(\alpha;\kk) \mathbf u , \ \  \mathscr L_{\mathbf a } \mathbf u ( z ) = e^{ i \langle  \mathbf{a}, \kk 
\rangle}  \mathbf u ( z ) , \ \ \kk \in \mathbb C/\Lambda^* . 
\end{equation}
The spectrum is symmetric with respect to the origin due to the chiral symmetry (cf.~\eqref{eq:defW}) and we index it as
\begin{equation}
\label{eq:eigs} 
\begin{gathered} \{ E_{ j } ( \alpha, \kk ) \}_{ j \in  \mathbb Z^* } ,  \ \  E_{j } ( \alpha; \kk ) = - E_{-j} ( \alpha ; \kk ) ,\ \ \mathbb Z^*:= \ZZ\backslash \{0\} 
\\ 0 \leq E_1 ( \alpha; \kk ) \leq E_2 ( \alpha; \kk ) \leq \cdots , \ \  E_1 ( \alpha; 0 ) =  E_1 ( \alpha; -i ) = 0. 
\end{gathered} \end{equation}
Each $E_j(\cdot;\kk),\, \kk\in \mathbb C/\Lambda^*$ is called a \emph{energy band} of the operator $H(\alpha)$. The points $ 0 , -i \in \mathbb C/\Lambda^* $
are called the {\em Dirac points} in the physics literature. The first two bands $E_{\pm 1}(\alpha;\kk)$ are pinned at zero energy at these two points for all $\alpha\in\CC$ by Proposition \ref{prop:prot} and \ref{prop:prot-1}, and the corresponding eigenstates are called \emph{protected states}.

For $E( \alpha ; \kk )$ satisfies equation \eqref{eq:FL_ev} of $H(\alpha)$, we define the following set of {\em magic angles}:
\begin{equation}
\label{eq:defA}  \mathcal A := \{ \alpha \in \mathbb C : \forall \, \kk \in \mathbb C/\Lambda^* , \ \
E_1 ( \alpha; \kk ) \equiv 0  \}. 
\end{equation}
And $\alpha\in \mathcal{A}$ is called {\em simple} if $E_2 ( \alpha; \kk_0 )>0$ for some $k_0\in \mathbb C/\Lambda^*$ and {\em multiplicity} of $\alpha\in \mathcal{A}$ is defined to be the number of $j$'s ($j>0$) such that $E_j ( \alpha, \kk ) \equiv 0$. It has been shown in \cite{yang2023flat} (using results of \cite{becker2020mathematics}) that $\mathcal{A}$ is a discrete non-empty set independent of $n$. 

In this paper, we study band structures of the TMG Hamiltonian $H(\alpha)$ in \eqref{eq:defH} for all $\alpha\in\CC$ near the symmetry points $K =\bo,-\bi$. We prove the following 
\begin{theo}
\label{thm:main}
    There exists a discrete set $\mathcal{B}\subset \CC$, such that the following trichotomy holds:  
    \begin{enumerate}
        \item \label{part1} for $\alpha\in \CC\backslash (\mathcal{A}\cup\mathcal{B})$ and $\kk$ close to $K = \bo,-\bi$, the first two bands are given by 
        \begin{equation*}
            E_{\pm 1}(\alpha;\kk) = \pm c_n (\alpha) |\kk-K|^{n} + \mathcal{O}(|\kk-K|^{n+1}), \quad c_n>0;
        \end{equation*}
        \item \label{part2} for $\alpha\in \mathcal{A}$, the first two bands are flat, i.e.,
        \begin{equation*}
            E_{\pm 1}(\alpha;\kk) \equiv 0 \quad \text{for all } k\in \mathbb C/\Lambda^* ;
        \end{equation*}
        \item \label{part3} for $\alpha\in \mathcal{B}$, and $\kk$ close to $K = \bo,-\bi$, the first four bands are given by
        \begin{gather*}
            E_{\pm 1}(\alpha;\kk) = \pm \Tilde{c}_{n-1} (\alpha) |\kk-K|^{n-1} + \mathcal{O}(|\kk-K|^{n}), \quad \Tilde{c}_{n-1}>0; \\
            E_{\pm 2}(\alpha;\kk) = \pm \Tilde{c}_1 (\alpha) |\kk-K| + \mathcal{O}(|\kk-K|^{2}), \quad \Tilde{c}_1>0.
        \end{gather*}
    \end{enumerate}
\end{theo}
    \begin{figure}[h]
        \centering
        \includegraphics[width=0.45\linewidth]{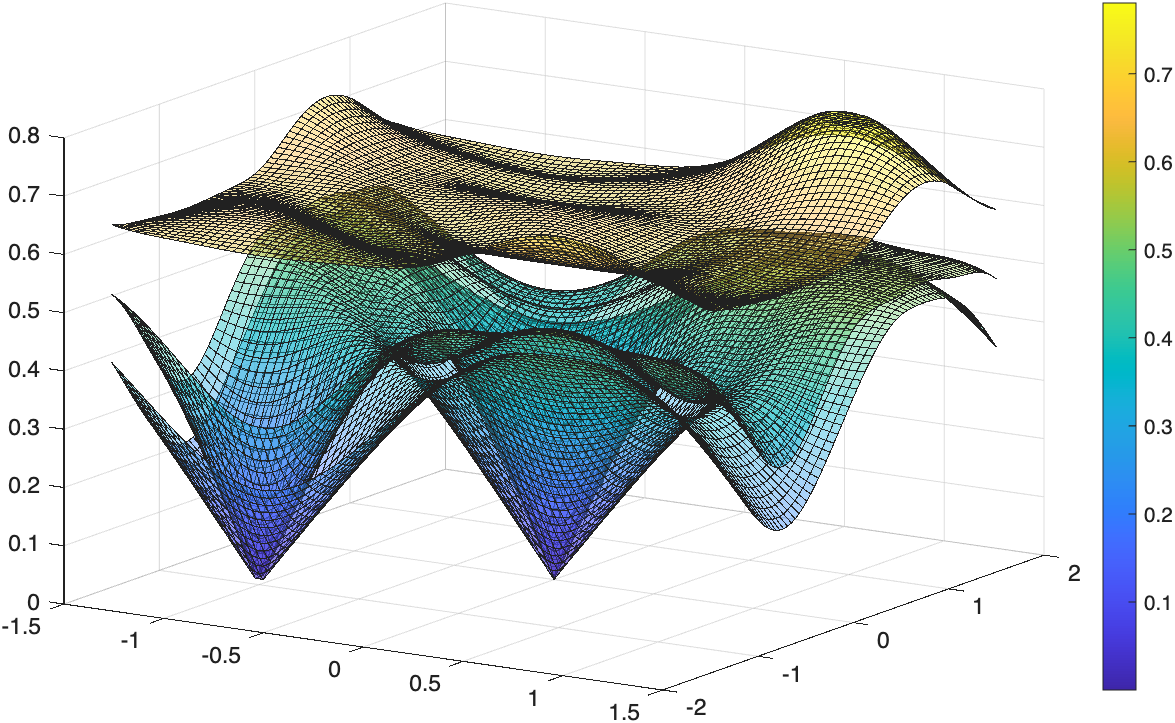}
        \quad
        \includegraphics[width=0.45\linewidth]{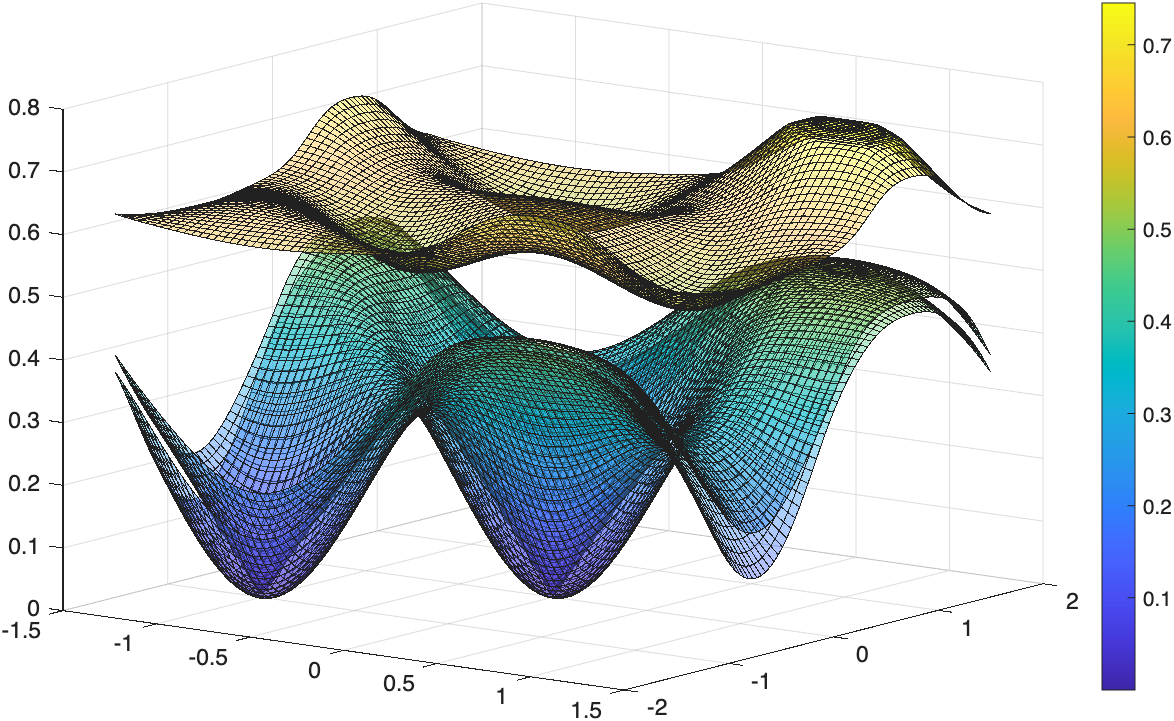}
        \caption{Band structures for $\alpha\in \mathcal{B}$. Only bands with non-negative energy are plotted here, as the bands $E_{\pm j}(k)$ are symmetric with respect to zero energy (cf.~Remark~\ref{rem:sym}). Left: when $n=2$, there are two Dirac cones at $K = 0, -i$ respectively, corresponding to the two bands $E_1(k), E_2(k)$. Right: when $n\geq 3$, there is a Dirac cone in $E_2(k)$ and an order $n-1$ tangential band in $E_1(k)$ at $K = 0, -i$ respectively.}
        \label{fig:cones}
    \end{figure}
    \begin{figure}[h]
        \centering
        \includegraphics[width=0.48\linewidth]{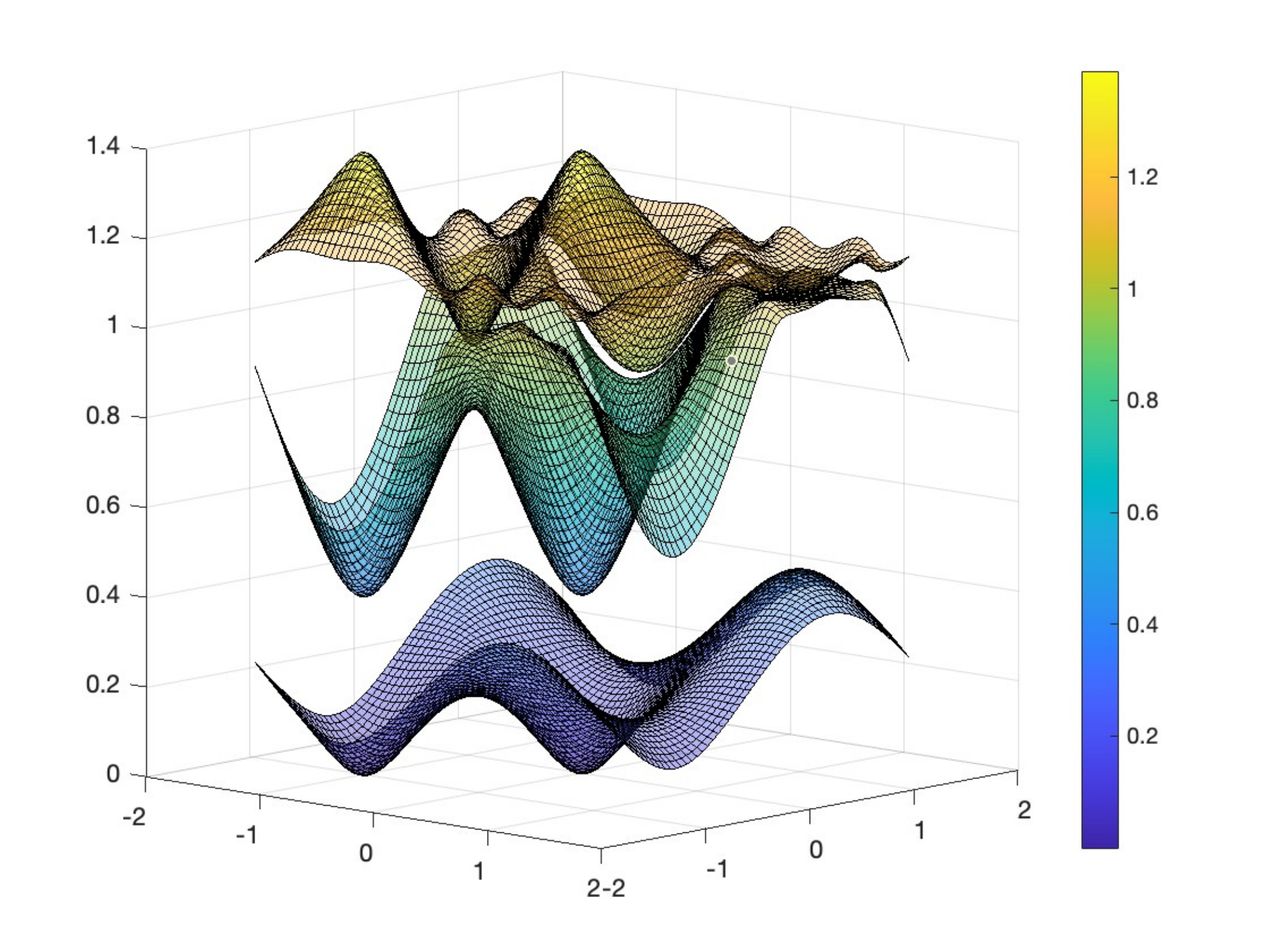}
        \quad
        \includegraphics[width=0.48\linewidth]{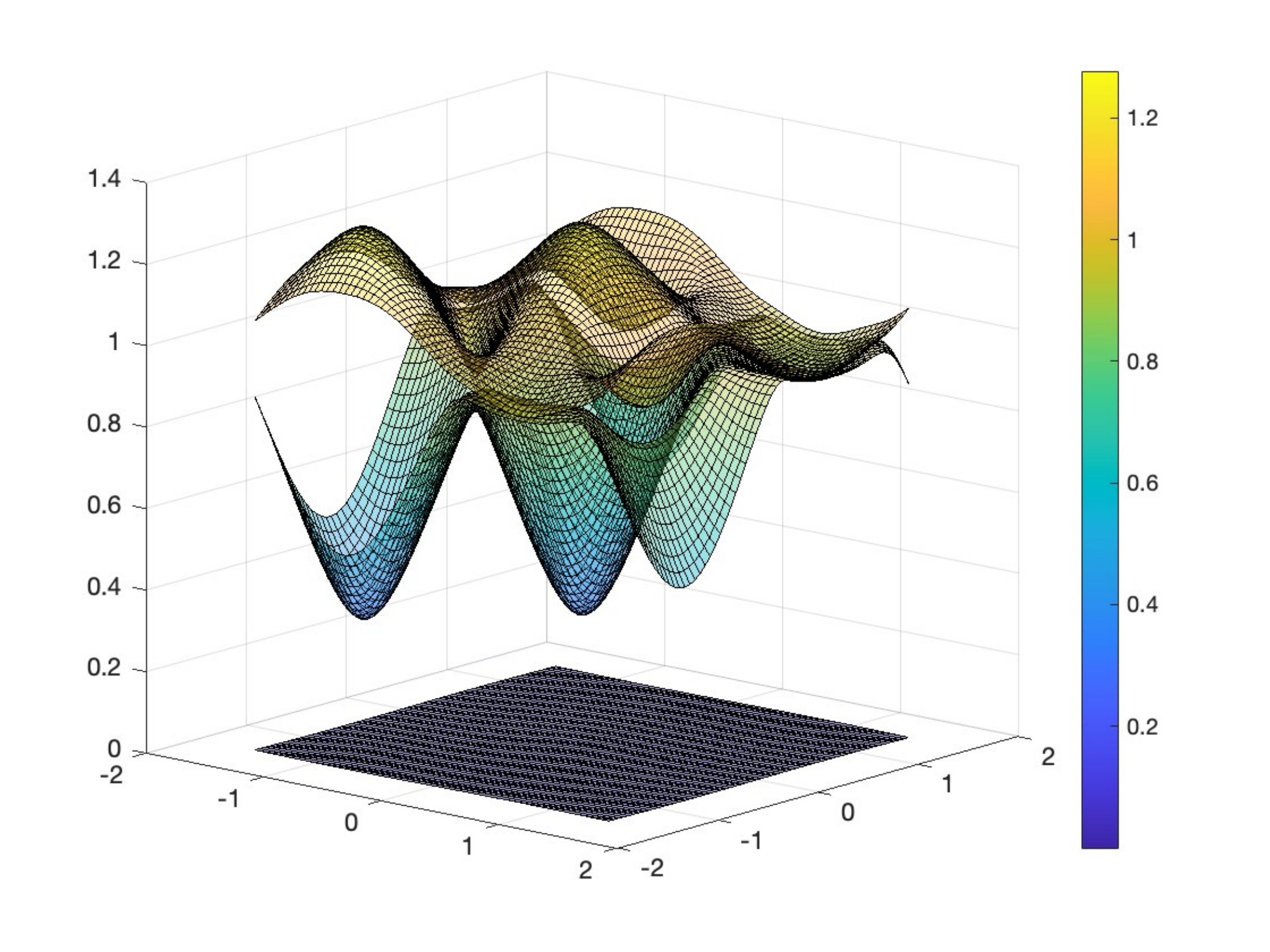}
        \caption{Three (non-negative) energy bands $E_j(k),\, j=1,2,3$ of $H(\alpha)$. Left: tangential band crossing $E_1(k)$ of order $n$ for $\alpha\notin \mathcal{A}\cup\mathcal{B}$. Right: flat band $E_1(k)\equiv 0$ for $\alpha\in \mathcal{A}$.}
        \label{fig:tangential}
    \end{figure}
This type of trichotomy (with $n=2$) was first observed in a scalar
model of flat bands by Dyatlov--Zeng--Zworski \cite{dyatlov2025scalar}, where in addition it was related to the structure of a rank-$2$ holomorphic vector bundle constructed from local solutions to $D(\alpha) u = 0$. 

Part \eqref{part1} of Theorem \ref{thm:main} shows that the generic band crossing in the TMG Hamiltonian $H(\alpha)$ is tangential rather than conic, where the band crossing is exactly of order $n$ (cf.~Figure~\ref{fig:tangential}). This is in contrast with the generic existence of Dirac cones in a single layer of graphene \cite{fefferman2012honeycomb,berkolaiko2018symmetry} and in twisted bilayer graphene \cite{zworski2023survey,becker2024dirac}. Part \eqref{part2} of Theorem \ref{thm:main} is proved in \cite[Theorem~1]{yang2023flat}; we include it here for completeness. The existence of Dirac cones (cf.~Figure~\ref{fig:cones} and part \eqref{part3}) in the TMG model \eqref{eq:defH} is new, and to the best of our knowledge, was not previously known in either theoretical or experimental physics. 

\smallsection{A numerical prediction for experiments}
Based on numerical computations of the set $\mathcal{B}$ (see Figure~\ref{fig:new-spec} and Table~\ref{tab:angles}) using equation \eqref{eq:Tk}, the twisting angle at which the band structures exhibit Dirac cones is predicted to be $\theta\simeq 0.44^\circ$, in comparison with the magic angle $\theta\simeq 1.1^\circ$ (corresponding to $\alpha_1 \simeq 0.58566$), where a flat band appears.

\begin{table}[h]
    \centering
    \begin{tabular}{rclcc}
\multicolumn{1}{c}{$k$} & &
\multicolumn{1}{c}{$\alpha_k$} &
 & $\beta_{k}$ \\[2pt] \hline
1  &\ &   \phantom{0}0.58566 && 1.45282            \\
2  &&     \phantom{0}2.22118  && 3.35798        \\
3  &&     \phantom{0}3.75140 && 4.64420        \\
4  &&     \phantom{0}5.27649   && 5.69075        \\
5  &&     \phantom{0}6.79478   && 7.50646        \\
6  &&     \phantom{0}8.31299    && 9.39597        \\
\end{tabular}
    \caption{Numerical values of the first six $\alpha\in \mathcal{A}$ and $\beta\in\mathcal{B}$.}
    \label{tab:angles}
\end{table}

\smallsection{Physical background and related works} 
We conclude the introduction by discussing its physical relevance and some other related works in mathematics. 

Two-dimensional \emph{moir\'e materials} emerge from the superposition of multiple slightly misaligned periodic layers (e.g.~honeycomb lattice for graphene), introducing a larger-scale periodic pattern known as a moiré superlattice. It gives rise to various remarkable physical phenomena, such as strongly correlated electrons \cite{cao2018correlated}, the quantum anomalous Hall effect \cite{PhysRevX.1.021014,ledwith2020fractional,xie2021fractional}, and superconductivity \cite{cao2018unconventional,park2021tunable}. One of the most widely studied and successful models characterizing physical properties of twisted bilayer or multilayer graphene is the Bistritzer--MacDonald (BM) \cite{bistritzer2011moire} Hamiltonian, which accurately predicts magic angles of twisted bilayer graphene (TBG). The chiral limit \eqref{eq:defH} of BM Hamiltonian for TBG was introduced by Tarnopolsky--Kruchkov--Vishwanath \cite{tarnopolsky2019origin}, which gives rise to a completely flat band $E_1(\alpha,k)=0$ for all $k\in\mathbb C$ and for certain twisting parameters $\alpha\in\mathcal{A}$. Such a flat band at zero energy corresponds to an eigenvalue of the chiral BM Hamiltonian with infinite multiplicity. Mathematical studies of moir\'e materials have since become very active and fruitful. 

Mathematical study of the chiral model of TBG was initiated by the works of Becker--Embree--Wittsten--Zworski \cite{becker2021spectral,becker2020mathematics}, who provided a spectral characterization of magic angles and explained the exponential squeezing of bands, and of Watson--Luskin \cite{watson2021existence}, who showed the existence of the first magic angle. A series of works by Becker--Humbert--Zworski \cite{becker2022fine,becker2022integrability,becker2023degenerate} investigates trace formulas, the existence of generalized magic angles, the existence and properties of degenerate magic angles, and topological features of energy bands in TBG. Becker--Oltman--Vogel \cite{becker2023magic,becker2024absence} studies the effects of randomness, showing that small magic angles can be destroyed by random perturbations. The effect of magnetic field have also been investigated: the works of Becker--Zworski \cite{becker2023dirac}, Becker--Kim--Zhu \cite{zhu2024magnetic}, and Becker--Zhu \cite{zhu2025spectral} investigate bands of chiral TBG under various types of magnetic fields.
The works of Hitrik–Zworski and Tao–Zworski \cite{hitrik2025classically} study classically forbidden regions for eigenstates.
The existence of Dirac cones away from magic angles was proved by Tao and  the second author in \cite[Appendix]{zworski2023survey}.
The chiral model of TMG was established by Ledwith--Vishwanath--Khalaf \cite{ledwith2022family} and Wang--Liu \cite{wang2022hierarchy}, where arbitrary Chern numbers were found. 
More recently, Becker--Humbert--Wittsten--Yang \cite{becker2025chiral} studied the band structure and topological features of twisted trilayer graphene, and the second author \cite{yang2023flat} analyzed the magic angles of TMG along with the topology of the associated flat bands.
On the many-body side, the ground states of the chiral TBG have been completely characterized by the works of Becker--Lin--Stubbs \cite{becker2025exact,stubbs2024hartree} and Stubbs--Ragone--MacDonald--Lin \cite{stubbs2025many}.
For an overview of mathematical results on the chiral model, see the survey by Zworski \cite{zworski2023survey}.

For the more general BM Hamiltonian as well as other related models, the mathematical derivation has been established by Watson--Kong--MacDonald--Luskin \cite{watson2023bistritzer}, Quinn--Kong--Luskin--Watson \cite{quinn2025higher} using wave packet methods, and Cancès--Garrigue--Gontier \cite{PhysRevB.107.155403} using the density function theory. The work of Cancès--Meng \cite{cances2023semiclassical} studies the asymptotics of density of states using semiclassical analysis. The work of Becker--Zworski \cite{becker2024chiral} studies the band structure of BM Hamiltonian for small coupling parameters. In joint work with Becker et~al.~\cite{becker2024dirac}, the second author proved the existence of Dirac cones and magic angles in the full BM model. A mathematically more elegant scalar model was proposed by Galkowski--Zworski \cite{galkowski2023abstract}, where abstract conditions on the existence of flat bands are established. The scalar model is being investigated further by Dyatlov--Zeng--Zworski \cite{dyatlov2025scalar}, where they show the asymptotic quantization rule of generalized magic angles (cf.~\cite[Open Problem~1]{zworski2023survey}) by carefully studying the WKB structure of the solution.  

\smallsection{Structure of the paper}
In Section \ref{sec:TML}, we recall some basic facts about the symmetries and symmetry-protected states of the TMG Hamiltonian \eqref{eq:defH}. Section \ref{sec:jordan} establishes the existence of a Jordan block structure (cf.~Theorem \ref{thm:jordan}) away from a discrete set $\mathcal{B}$ of values of $\alpha$. Building on this, Section \ref{sec:tangential} shows that the band crossing is exactly of order $n$ for $\alpha \in \mathbb{C} \backslash (\mathcal{A} \cup \mathcal{B})$. Finally, in Section \ref{sec:cone}, we complete the picture of Theorem~\ref{thm:main} by proving that the band structure exhibits Dirac cones when $\alpha \in \mathcal{B}$.

\smallsection{Acknowledgment}
The authors would like to thank Semyon Dyatlov, Henry Zeng and Maciej Zworski for inspiring discussion regarding the protected Jordan block structures in the scalar model \cite{zworski2025scalar}, as well as for providing numerical help. The authors would also like to thank Yuan Cao, Patrick Ledwith, Zhongkai Tao and Jie Wang for their interest in the project and helpful discussions. MY acknowledges support from the AMS-Simons Travel Grant and NSF grant DMS-2509989, as well as partial support from the DARPA grant AIQ-HR001124S0029. BL is partially supported by the UC Berkeley Summer Undergraduate Research Fellowships (SURF) program.

\section{TMG Hamiltonians}
\label{sec:TML}
In this section, we recall symmetries of the Hamiltonian $H(\alpha)$. The symmetry group commutes with the Hamiltonian and split the functional space into irreducible representations. This gives rise to the so-called protected states.

\subsection{Symmetries of the TMG Hamiltonian}
We briefly recall the symmetries of the Hamiltonian $H(\alpha)$ discussed in \cite{yang2023flat}. Recall that the lattice is given by
\begin{equation}
\label{eq:defGam3}
 \Lambda = \tfrac43 \pi i \omega ( \ZZ \oplus \omega \ZZ)  
\end{equation}
We first define a twisted translation on $L_{\loc}^2 ( \CC ; \CC^{2n} )$
\begin{equation}
\label{eq:translation}
    {\mathscr L}_{\mathbf a } u(z) :=  \diag_n\left\{\begin{pmatrix}  \omega^{a_1+a_2} & 0 \\ 0 & 1 \end{pmatrix}  \right\}  u(z+ \mathbf a) ,  \ \ 
\mathbf a = \tfrac 43 \pi i \omega ( a_1 + \omega a_2 ) \in \Lambda,
\end{equation}
where $\mathrm{diag}_n\{(*)\}$ denotes block diagonal matrices with $n$ blocks. The operator ${\mathscr L}_{\mathbf a }$ satisfies 
\begin{equation}
\label{eq:defLa0} 
\mathscr L_{\mathbf a } D_n ( \alpha )  = D_n ( \alpha )  \mathscr L_{\mathbf a } , \ \ \ 
\end{equation}
We can further extend the action of $ \mathscr L_{\mathbf a } $ to 
$ L_{\loc}^2 ( \CC ; \CC^{4n} ) $ 
block-diagonally and it yields 
\begin{equation}
\mathscr L_{\mathbf a} H ( \alpha ) = H ( \alpha ) \mathscr L_{\mathbf a }, \ \ 
\mathbf a  \in \Lambda.
\end{equation}
Hence, after a conjugation by $e^{i\langle z,k \rangle}$, the spectrum of $H(\alpha)$ under the Floquet condition \eqref{eq:FL_ev} for $\kk \in \mathbb C/\Lambda^*$ is equivalent to
\begin{equation}
\label{eq:FL-Hk}  
H_k(\alpha) \mathbf u = E(\alpha;\kk) \mathbf u , \ \  \mathscr L_{\mathbf a } \mathbf u ( z ) =  \mathbf u ( z ) , \quad \text{for all } \mathbf{a}\in\Lambda, 
\end{equation}
where
\begin{equation}
\label{eq:Hk}
    H_k(\alpha):= 
    \begin{pmatrix} 0 & D(\alpha)^*+\Bar{k} \\
  D ( \alpha) +k &  0 
  \end{pmatrix},
\end{equation}

Now we consider the rotational symmetry. The second identity in \eqref{eq:symmU} implies that
$$  [ D ( \alpha )  u ( \omega \bullet ) ] ( z )  = \bar \omega 
[ D ( \alpha )  u ] ( \omega z ) .$$
Therefore, we define the twisted rotation operator on $L_{\loc}^2 ( \CC ; \CC^{2n} )$ as
\begin{equation}
    \label{eq:c2n}
    \mathscr C_{2n}u (z) := J u(\omega z), \quad  J := \diag(1,1, \bar{\omega}, \omega, \cdots, \bar{\omega}^{n-1}, \omega^{n-1})
\end{equation}
so that
\begin{equation}
     D_n ( \alpha ) \mathscr C_{2n}  = \bar \omega \mathscr C_{2n} D_n ( \alpha ).
\end{equation}
By a slightly abuse of notation, this can be extended to $ L_{\loc}^2 ( \CC ; \CC^{4n} )$ such that 
\begin{gather}
    \label{eq:c4n}
    \mathscr C_{4n} H(\alpha) = H(\alpha) \mathscr C_{4n}, \quad \mathscr C_{4n} u ( z ) := \begin{pmatrix} J & 0  \\ 0 &  \bar \omega J  \end{pmatrix} u ( \omega z ).
\end{gather}

We record an additional action and chiral symmetry involving $H(\alpha)$. 
\begin{equation}
\label{eq:defW}
\begin{gathered}   H(\alpha) = - \mathscr W H(\alpha) \mathscr W^*, \ \ \ \mathscr W := \begin{pmatrix}
\id & 0 \\
0 & -\id \end{pmatrix} ,
\ \ \mathscr W  \mathscr C_{4n}  = 
 \mathscr C_{4n}  \mathscr W ,
 \ \ \mathscr L_{\mathbf a } \mathscr W = \mathscr W \mathscr L_{\mathbf a },
\end{gathered}
\end{equation}
This implies that the spectrum of $H(\alpha)$ is even. 
\begin{rem}
\label{rem:sym}
    For the Bloch transformed Hamiltonian $H_{k}(\alpha)$ in \eqref{eq:Hk}, we also have 
    $$H_{k}(\alpha)  = - \mathscr W H_{k}(\alpha) \mathscr W^*.$$
    This implies that for each eigenvalue $E_j(k)\geq0$ of $H_{k}(\alpha)$, there exists another eigenvalue $E_{-j}(k)\leq0$ such that $E_{-j}(k) = -E_{j}(k)$.
\end{rem}

\subsubsection{Group actions on functional spaces}
Following \cite{becker2020mathematics}, since $ \mathscr C_{4n} \mathscr L_{\mathbf a } = \mathscr L_{\bar \omega \mathbf a } \mathscr C_{4n}$, we can combine the two actions into a unitary group action that
commutes with $H(\alpha)$:
\begin{equation}
\label{eq:defG}  
\begin{gathered}  G := \Lambda \rtimes \ZZ_3 , \ \ 
  \ZZ_3 \ni k : \mathbf a \to \bar \omega^k  \mathbf a , \ \ \ ( \mathbf a , k ) \cdot ( \mathbf a' , \ell ) = 
( \mathbf a + \bar \omega^k \mathbf a' , k + \ell ) ,
\\  ( \mathbf a, \ell ) \cdot u = 
\mathscr L_{\mathbf a } \mathscr C_{4n}^\ell u ,\quad  u \in L^2 ( \CC; \CC^{4n} ) \ \
\end{gathered}
\end{equation}
Taking a quotient by $ 3\Lambda $, we obtain a finite group acting 
unitarily on $ L^2 ( \CC/3\Lambda; \CC^{4n} ) $ and commuting with $ H(\alpha) $:
\begin{equation}
\label{eq:defG3}
G_3 :=  G/3\Lambda = \Lambda/3\Lambda \rtimes \ZZ_3 \simeq \ZZ_3^2 \rtimes \ZZ_3. 
\end{equation}
Restricting to the components, $G_3$ acts on $ L^2 ( \CC/3\Lambda; \CC^{2n} ) $ as well and we 
use the same notation for those actions. For $\bullet = 2n, 4n$, we define the following function spaces using the group action $G_3$:
\begin{equation}
\label{eq:Lkp}
L^2_{\kappa,p} ( \CC; \CC^{\bullet} ):= \left\{{u}\in L^2_{\loc}(\CC; \CC^{\bullet}): \  \mathscr L_{\mathbf a } \mathscr C^{\ell}  u = e^{i\langle \mathbf a, \kk \rangle} \bar \omega^{\ell p}   u\right\}, \quad \ell,p,\kappa\in\ZZ_3, \quad k = -i\kappa.
\end{equation}
The commutativity of $\mathscr{L}_{\mathbf{a}}$ and $\mathscr C_{4n}$ with $H(\alpha)$ and the chiral symmetry $\mathscr W$ yields that
\begin{prop}
  \label{prop:specH}
The operator $H(\alpha): L^2 ( \CC; \CC^{4n} ) \to L^2 ( \CC; \CC^{4n} )$ is
an unbounded self-adjoint operator with the domain given by 
$ H^1 ( \CC; \CC^{4n} ) $ and 
\[ \Spec_{ L^2 ( \CC ) } H(\alpha)  = 
- \Spec_{ L^2 ( \CC ) } H(\alpha). \]
The same conclusions holds when $ L^2 (\CC )$ is replaced by $ L^2_{\kappa,p}( \CC; \CC^{4n} )$, where the spectrum is discrete. 
\end{prop}

It is also convenient to consider the space
\begin{gather}
\label{eq:Lk}
    L^2_{k} ( \CC; \CC^{4n} ):= \left\{{u}\in L^2_{\loc}(\CC; \CC^{4n}): \  \mathscr L_{\mathbf a }   u = e^{i\langle \mathbf a, \kk \rangle}   u\right\}, \\ 
    \label{eq:decomp}
    L^2_{k} ( \CC; \CC^{4n} ) = \bigoplus_{p\in\ZZ_3} L^2_{\kappa,p} ( \CC; \CC^{4n} ), \quad k= -i\kappa. 
\end{gather}

\subsection{Protected states at zero energy}
\label{sec:protect}
We recall the result of the existence of protected zero eigenstates of the operator $H(\alpha)$. The kernel of $H(\alpha)$ at $\alpha=0$ is given by 
$$\ker_{L^2 ( \CC/3\Lambda ; \CC^{4n}) } H(0)= \{{e}_1, {e}_{2n}, {e}_{2n+2}, {e}_{4n-1}\},$$ 
where the set $ \{{e}_j\}_{j=1}^{4n} $ forms the standard basis of $ \CC^{4n} $. In particular, we have
\begin{equation}
    \label{eq:ps}
    e_1 \in L^2_{ {1,0}}  ( \CC; \CC^{4n} ) , \ \  
    e_{2n} \in L^2_{ {0,[1-n]}}  ( \CC; \CC^{4n} ) , \ \ 
    e_{2n+2} \in L^2_{ {0,1}}  ( \CC; \CC^{4n} ), \ \  
    e_{4n-1} \in L^2_{ {1,[n]}}  ( \CC; \CC^{4n} ). 
\end{equation}
Since $ \mathscr W $ commutes with the action of $ G_3 $, the spectra of $ H( \alpha)\rvert_{L^2_{k,p}( \CC; \CC^{4n} )}$ are symmetric with respect to $0$ (see Proposition \ref{prop:specH}). As
\begin{equation}
    \dim \ker_{L^2_{\kappa,p}( \CC; \CC^{4n} )} H(0) = 1
\end{equation}
for $L^2_{\kappa,p}( \CC; \CC^{4n} )$ in \eqref{eq:ps}, it follows that such $ H( \alpha) \rvert_{L^2_{k,p}( \CC; \CC^{4n} )} $ has an eigenvalue at $0$ for any $\alpha\in \CC$ as the operator $H(\alpha)$ is holomorphic for $\alpha\in \CC$. Since 
\begin{equation}
\label{eq:kerH-decomp}
    \begin{split}
        \ker_{{L^2_{\kappa,p}}( \CC; \CC^{4n} )} H(\alpha) = 
 \ker_{{L^2_{\kappa,p}( \CC; \CC^{4n} )} }  D_n(\alpha) \oplus \{ 0_{\CC^{2n} } \}
  + \{ 0_{\CC^{2n} } \} \oplus \ker_{{L^2_{\kappa,p}( \CC; \CC^{4n} )} } D_n(\alpha)^*,
    \end{split}
\end{equation} 
we obtained the following result about symmetry protected eigenstates at $0$:
\begin{prop}
\label{prop:prot}
For all $ \alpha \in \CC $, 
\[ \dim \ker_{ L^2_{{ 1,0} }( \CC; \CC^{2n} ) }  D_n(\alpha)  \geq 1 , \quad  
{e}_{2n} \in \ker_{ L^2_{{ 0,[1-n]} }( \CC; \CC^{2n} )}  D_n(\alpha).
\]
\end{prop}
In addition, by \eqref{eq:c4n} and \eqref{eq:kerH-decomp}, we have
\begin{prop}
\label{prop:prot-1}
For all $ \alpha \in \CC $, 
\[ \dim \ker_{ L^2_{{0,0} }(\CC; \CC^{2n})}  D_n(\alpha)^*  \geq 1 , \quad  
{e}_{2n-1} \in \ker_{ L^2_{{1,[n-1]}}( \CC; \CC^{2n} )}  D_n(\alpha)^*.
\]
\end{prop}

\section{Generic existence of Jordan block structure}
\label{sec:jordan}
In contrast to TBG, we shall see that the TMG model has a generically protected Jordan block structure, which leads to new phenomena in the band structure of the TMG Hamiltonian. From now on, we focus on the Dirac point $\kappa=0$ (i.e.,~$\bk=0$), the other Dirac point $\kappa=1$ (i.e.,~$\bk =-i$) is similar. For notational convenience, we define
$$L^2_{\kappa,p}:= L^2_{\kappa,p} ( \CC; \CC^{2n} ), \quad L^2_{k}:= L^2_{k} ( \CC; \CC^{2n} ).$$

Using the fact that $\Spec_{L^2_0}D_n(\alpha) = 3\Lambda^* + \{0, i\}$ for $\alpha\notin\magic$ (cf.~\cite[equation (3.12)]{yang2023flat}), we define the spectral projection 
\begin{equation}
    \Pi (\alpha) := -\frac{1}{2\pi i}\oint_{\gamma_0} (D_n(\alpha) - z)^{-1}\vert_{L^2_0} \,dz, \quad \alpha\notin\magic,
\end{equation}
where the contour $\gamma_0$ is taken to be near 0 independent of $\alpha$.

Note that for $\alpha=0$, the operator 
\begin{equation}
    \label{eq:D-map}
    D_n(\alpha): L^2_{0,j}\to L^2_{0,j-1}, \quad j\in\ZZ_3
\end{equation} 
has a Jordan block structure. That is, for $e_2\in \coker_{L^2_{0,0}}D_n(0)$ and $e_{2n}\in \ker_{L^2_{0,[1-n]}}D_n(0)$, we have 
\begin{equation}
\label{eq:jb0}
    D_n(0)e_{2k} = t_k e_{2k+2} \in L^2_{0,[-k]}, \quad k= 1,2, \cdots , n-1.
\end{equation}
We would like to study if such Jordan block structure of $D_n(\alpha)$ persists for $\alpha\in\CC$. In this section, we prove the following 
\begin{theo}
\label{thm:jordan}
There exists a discrete set $\mathcal{B}\subset \CC$, such that for any $\alpha\in \CC\backslash (\mathcal{B}\cup \mathcal{A})$, we have $\rank \Pi(\alpha) = n$, and there exist $u_k\in L^2_{0,[k-n]}, k = \{1, 2, \cdots, n\}$  such that
\begin{gather}
    \label{eq:ker-coker}
    \coker_{ L^2_{{ 0,0} } }  D_n(\alpha)  = \CC u_n, \quad \ker_{ L^2_{{ 0,[1-n]} } }  D_n(\alpha)  =  \CC u_1, \\
    \label{eq:chain}
    D_n(\alpha) u_{k+1} = u_{k}, \quad u_1 =  e_{2n}.
\end{gather}
\end{theo}

As the projection $\Pi(\alpha)$ is holomorphic in $\CC\backslash\magic$ with $\rank \Pi(0) = n$, we have  
\begin{equation}
\label{eq:rank-2}
    \rank \Pi(\alpha) \equiv \tr \Pi(\alpha) \equiv n, \quad \alpha\in \CC\backslash\magic,
\end{equation}
which implies the generalized kernel of $D_n(\alpha)\rvert_{L^2_0}$ is $n$-dimensional for $\alpha\in \CC\backslash\magic$. Hence, it implies the first part of Theorem \ref{thm:jordan}. For $\alpha=0$, we define the orthogonal complement of $\ker_{ L^2_{{ 0,[1-n]} }}  D_n(0)$ and $\coker_{ L^2_{{0,0} }}  D_n(0)$: 
\begin{gather}
\label{eq:newspace}
    \prescript{\sharp}{}{L^2_{{0,j}}}:= \{f\in L^2_{{0,j}}: \langle f, e_{2n} \rangle = 0 \}, \quad 
    \prescript{\flat}{}{L^2_{{0,j}}}:= \{f\in L^2_{{0,j}}: \langle f, e_{2} \rangle = 0 \}.
\end{gather}
We introduce a useful lemma on the (almost) invertibility of $D(\alpha)\rvert_{L^2_{0,j}}$.
\begin{lemm}
\label{lem:invert}
There exist discrete sets $\mathcal{B}_j, j\in \ZZ_3$ such that the operator 
$$ D_n(\alpha): \prescript{\sharp}{}{L^2_{{0,j}}} \to \prescript{\flat}{}{L^2_{{0,j-1}}}$$ 
is invertible for $\alpha\notin \mathcal{B}_j$.     
\end{lemm}

\begin{proof}
Recall that the operator $D_n(\alpha): L^2_{0}\to L^2_{0}$ is Fredholm with index zero with protected states 
\[{e}_{2n} \in \ker_{ L^2_{{ 0,[1-n]} }}  D_n(\alpha), \quad \dim \coker_{ L^2_{{0,0} }}  D_n(\alpha) \geq 1.\]
In particular, when $\alpha=0$, we have
\[\dim \ker_{ L^2_{{ 0,[1-n]} }}  D_n(0) = \dim \coker_{ L^2_{{0,0} }}  D_n(0)=1.\]
By equation \eqref{eq:decomp} and \eqref{eq:D-map}, we conclude that the operator 
\[D_n(\alpha): L^2_{0,j}\to L^2_{0,j-1}, \quad j\in\ZZ_3\]
is Fredholm with possible kernel and cokernel containing the protected states. Hence, the operator 
\begin{gather}
    D_n(0): \prescript{\sharp}{}{L^2_{{0,j}}} \to \prescript{\flat}{}{L^2_{{0,j-1}}}, \quad j\in \ZZ_3
\end{gather}
is invertible\footnote{Note that for $j\not\equiv [1-n]$, $\prescript{\sharp}{}{L^2_{{0,j}}} = L^2_{{0,j}}$ and for $j\not\equiv 0$, $\prescript{\flat}{}{L^2_{{0,j}}} = L^2_{{0,j}}$.}. Formally we shall write
\begin{gather}
    \label{eq:res-id}
    D_n(\alpha)\rvert_{\prescript{\sharp}{}{L^2_{{0,j}}}} = D_n(0)\rvert_{\prescript{\sharp}{}{L^2_{{0,j}}}} \left(\id + \alpha \big(D_n(0)\rvert_{\prescript{\sharp}{}{L^2_{{0,j}}}}\big)^{-1} V_n \right), \\ 
    \label{eq:Vn}
    V_n = \diag \left( \begin{pmatrix} 0  &    U(z) \\  
     U(-z) & 0   \end{pmatrix}, 0_{2\times 2}, \cdots \right),
\end{gather}
where
\begin{gather}
\label{eq:idK}
    \id + \alpha \big(D_n(0)\rvert_{\prescript{\sharp}{}{L^2_{{0,j}}}}\big)^{-1} V_n : \prescript{\sharp}{}{L^2_{{0,j}}} \to \prescript{\sharp}{}{L^2_{{0,j}}}. 
\end{gather}
The operator \eqref{eq:idK} is well-defined provided the multiplication operator $V_n$ maps the space $\prescript{\sharp}{}{L^2_{{0,j}}}$ into $\prescript{\flat}{}{L^2_{{0,j-1}}}$. This can be easily verified: the only non-trivial case is $j\equiv 1$. In such case, assume $V_n u = e_2$, then we see that by \eqref{eq:Vn} and \eqref{eq:symmU}, $u$ is not locally $L^2$-integrable. Hence, the identity \eqref{eq:res-id} holds. 

By \eqref{eq:res-id}, the operator $D_n(\alpha): {\prescript{\sharp}{}{L^2_{{0,j}}}} \to \prescript{\flat}{}{L^2_{{0,j-1}}}$ is invertible if and only if 
\begin{equation}
\label{eq:Tk}
    -\frac{1}{\alpha}\notin \Spec_{ \prescript{\sharp}{}{L^2_{{0,j}}} } D_n(0)^{-1} V_n.
\end{equation}
As the operator $ D_n(\alpha)^{-1} V_n: \prescript{\sharp}{}{L^2_{{0,j}}}\ \to \prescript{\sharp}{}{L^2_{{0,j}}}$ is compact, it has a discrete set of spectrum $\mathcal{B}_j$. This yields the invertibility of $ D_n(\alpha): \prescript{\sharp}{}{L^2_{{0,j}}} \to \prescript{\flat}{}{L^2_{{0,j-1}}}$ for $\alpha\notin \mathcal{B}_j$. 
\end{proof}

\begin{proof}[Proof of Theorem \ref{thm:jordan}]
Recall that the operator $D_n(\alpha): L^2_{0}\to L^2_{0}$ is an analytic family of Fredholm operators with index zero, and with protected states given by 
\[{e}_{2n} \in \ker_{ L^2_{{ 0,[1-n]} }}  D_n(\alpha), \quad \dim \coker_{ L^2_{{0,0} }}  D_n(\alpha) \geq 1.\]
An elementary linear algebra computation yields a chain of generalized eigenstates
$$D_{n}(\alpha)u_{m+1} = t_{n-m} u_{m}, \quad  m = 1,2,\cdots n-2, \quad u_{m} := e_{2(n+1-m)}.$$ 
This gives $n - 1$ generalized eigenstates directly. It remains to solve the equation
\begin{equation}
\label{eq:psi}
    D_{n}(\alpha) u_n = e_{4}
\end{equation}
As $e_4\in \prescript{\flat}{}{L^2_{{0,-1}}}$ trivially orthogonal to $e_2$, by Lemma \ref{lem:invert} and definition \eqref{eq:newspace}, for $\alpha\notin \mathcal{B}_0$ there exists an element $u_n\in \prescript{\sharp}{}{L^2_{{0,0}}}$ solves the equation \eqref{eq:psi}. Hence, combining with the fact that $\rank \Pi(\alpha)\equiv n$ for $\alpha\notin\magic$, the set $\mathcal{B} \coloneq \mathcal{B}_0\backslash \magic$ satisfies the theorem. 
\end{proof}

\begin{rem}
    Note that the set $\mathcal{B}$ is independent of $n$. This follows from the exactly same argument in the proof of \cite[Theorem~1]{yang2023flat}. 
\end{rem}

\begin{figure}[h]
    \centering
    \includegraphics[width=0.70\linewidth]{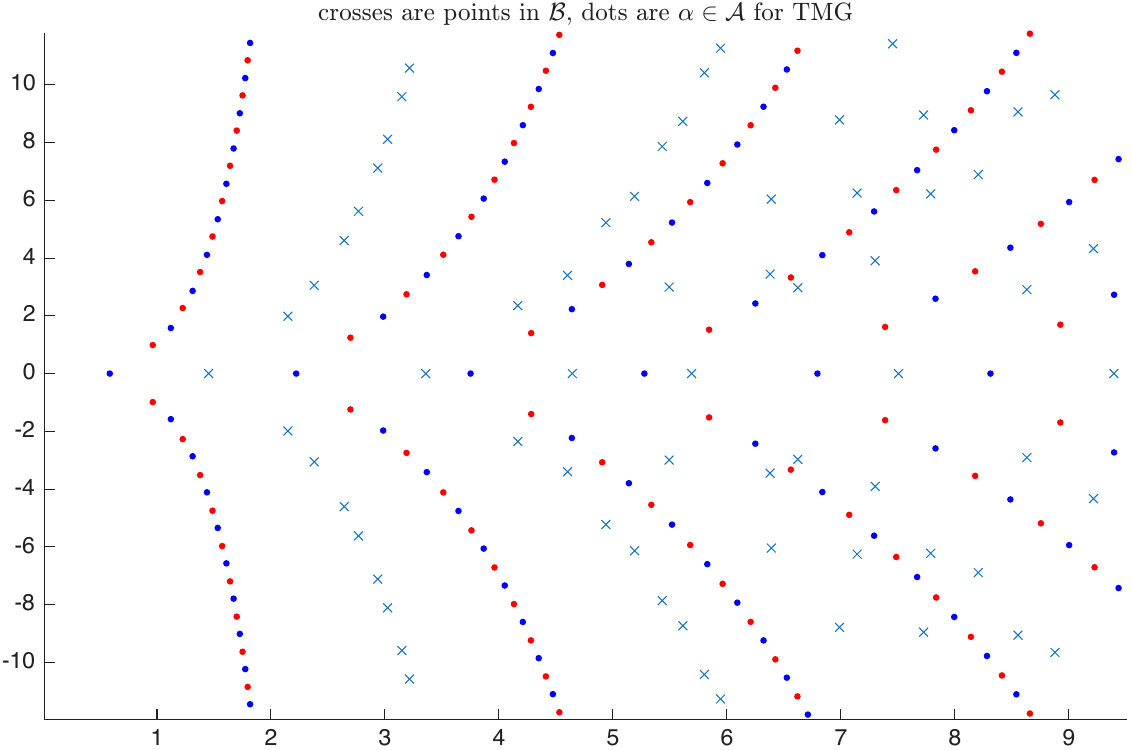}
    \caption{The set $\magic$ (dots) and $\mathcal{B}$ (crosses).}
    \label{fig:new-spec}
\end{figure}

From the proof of Theorem \ref{thm:jordan}, the holomorphic family of Fredholm operators  
\[ D_n(\alpha): \prescript{\sharp}{}{L^2_{{0,0}}} \to \prescript{\flat}{}{L^2_{{0,-1}}}\]
with index zero is not invertible for $\alpha\in \mathcal{B}_0$. Hence it must have a nontrivial kernel in $\prescript{\sharp}{}{L^2_{{0,0}}}$. For $\alpha\in \mathcal{B}\subset \mathcal{B}_0$, we have $\rank \Pi(\alpha) = n$ with the generalized kernel contains $\{u_1, \cdots, u_{n-1}\}$ satisfying \eqref{eq:chain} orthogonal to each other. Hence, we must have exactly
    \begin{equation}
        \dim \ker_{\prescript{\sharp}{}{L^2_{{0,0}}}}  D_n(\alpha) = \dim \coker_{\prescript{\flat}{}{L^2_{{0,-1}}}}  D_n(\alpha) = 1.
    \end{equation}

Therefore, combining with Proposition \ref{prop:prot} and \ref{prop:prot-1}, we have the following
\begin{corr}
\label{cor:dim}
   We have
   \begin{enumerate}
        \item for $\alpha \notin \mathcal A \cup \mathcal B$, $\dim \ker_{L^2_0} D_{n}(\alpha) = \dim \ker_{L^2_0} D_{n}(\alpha)^* = 1$;
        \item for $\alpha\in \mathcal B$, $\dim \ker_{L^2_0} D_{n}(\alpha) = \dim \ker_{L^2_0} D_{n}(\alpha)^* = 2$.
    \end{enumerate}
\end{corr}

\begin{rem}
    The dimension of kernel can be much more complicated for $\alpha\in \mathcal{A}$, depending on the form of the potential $U$ in \eqref{eq:defU}. In fact, the dimension is given by the multiplicity of magic angles. As we focus on band structures away from magic angles, we will not discuss this issue in detail. See \cite{becker2023degenerate,yang2023flat} for more details. 
\end{rem}

\section{Generic existence of higher order band touching}
\label{sec:tangential}
In this section, we show part \eqref{part1} of Theorem \ref{thm:main}. Namely, for $k$ near $K = \bo, -\bi$ where $D_n(\alpha)+K$ has a Jordan block structure of rank $n$, the first two bands are given by \[E_{\pm 1} (\alpha; k) = \pm c_n(\alpha) |\kk - K|^n + \mathcal{O}(|\kk - K|^{n + 1}), \quad c_n(\alpha)>0.\]
Again, we focus on the band structure near the Dirac point $K=\bo$ for notational simplicity. The computation near $K=-\bi$ is essentially the same. We employ the so-called Grushin problem framework. For a more detailed discussion on Grushin problems, we refer to \cite[Appendix~C]{dyatlov2019mathematical}.  We start by proving a useful lemma.

\begin{lemm}\label{lem:dichotomy}
    Let $P \colon X \to X$ be a Fredholm operator of index $0$ for a Hilbert space $X$. Assume $$\ker P = \mathrm{span}\,\{\varphi_1\}, \quad \ker P^* = \mathrm{span}\,\{\varphi^*_1\}.$$ 
    Suppose that the generalized kernels of $P$ and $P^*$ are of dimension $n$, with corresponding generalized kernel elements $\varphi_1, \dots, \varphi_{n}$ and $\varphi_1^*, \dots, \varphi_{n}^*$ such that 
    $$P\varphi_i = \varphi_{i-1}, \quad P^*\varphi_i^* = \varphi_{i-1}^*, \quad i = 2,\cdots, n.$$ 
    Then
    \begin{enumerate}
        \item for any $k,\ell$ such that $k + \ell \leq n$, we have $\langle \varphi_k^*, \varphi_\ell\rangle = 0$;
        \item there is a dichotomy that either the operator $P - z$ is invertible in some punctured disk around $0$, in which case 
        $$\langle \varphi^*_1, \varphi_{n}\rangle = \langle \varphi^*_{n}, \varphi_1\rangle \ne 0,$$ 
        or the operator $P - z$ is not invertible for all $z$ with
        $$\langle \varphi^*_1, \varphi_{n}\rangle = \langle \varphi^*_{n}, \varphi_1\rangle = 0.$$
    \end{enumerate}
\end{lemm}

\begin{proof}
    Part (1) follows from the identity
    \begin{equation}
        \langle \varphi_k^*, \varphi_\ell\rangle = \langle (P^*)^{n - k}\varphi_{n}^*, \varphi_\ell \rangle = \langle \varphi^*_{n}, P^{n-k}\varphi_{\ell}\rangle = 0, \quad \text{for } n-k\geq \ell.
    \end{equation}  
    Now we prove part (2). The dichotomy on the invertibility of $P-z$ follows from analytic Fredholm theory (see \cite[Theorem~C.8]{dyatlov2019mathematical}). If $P - z$ is invertible in some punctured disk around $0$, for $r>0$ sufficiently small, consider the rank $n$ spectral projector 
    \[\Pi = \frac{1}{2\pi i}\oint_{\partial B_r(0)} (\zeta - P)^{-1}\,\mathrm d{\zeta},\] 
    such that $\Pi^2 = \Pi$, $\Pi P = P \Pi$ and $\Pi P^n = P^n \Pi=0$. Hence we have
    \begin{equation}
    \label{eq:prop-Pi}
        \operatorname{ran}\Pi \subset \ker P^n, \quad \operatorname{ran}P^n \subset \ker \Pi.
    \end{equation} 
   Equation \eqref{eq:prop-Pi} yields that
   \begin{equation}
       \label{eq:proj}
       \Pi v = \sum a_{ij} \langle v, \varphi^*_i\rangle \varphi_j
   \end{equation} 
   with a rank $n$ matrix $(a_{ij})_{1\leq i,j\leq n}$.  
   Note that 
   $ \varphi_1 = \Pi \varphi_1 = \sum a_{ij} \langle \varphi_1, \varphi^*_i \rangle \varphi_j.$  
   Part (1) implies that for $i \leq n - 1$, we have $\langle \varphi_1, \varphi^*_i\rangle = 0$, which yields 
   $\varphi_1 = \sum a_{i,n}\langle \varphi_1 , \varphi^*_{n} \rangle \varphi_j.$  
   Hence, we must have $a_{1, n} \neq 0$ and $\langle \varphi_1, \varphi^*_{n} \rangle = \langle \varphi_n, \varphi^*_{1} \rangle \ne 0$.

    Conversely, if $P-z$ is not invertible for any $z$, we consider the Grushin problem 
    \begin{equation}
    \begin{pmatrix}
        P - z & R_- \\ R_+ & 0
    \end{pmatrix}
        \colon X \oplus \CC \to X \oplus \CC
    \end{equation}
    where $\varphi_1^*(R_-(1)) = 1$ and $R_+\varphi_1 = 1$. The above  Grushin problem is invertible for $z$ small, with inverse 
    \[ \begin{pmatrix}
        E(z) & E_+(z) \\ E_-(z) & E_{-+}(z)
    \end{pmatrix} \colon X \oplus \CC \to X \oplus \CC.\]  
    By \cite[Proposition C.3]{dyatlov2019mathematical}, we have 
    \[E_{-+}(z) = E_{-+}(0) + \sum_{m = 1}^\infty z^mE_-(0)E(0)^{m - 1}E_+(0)\] 
    As $E_{-+}(z)\equiv 0$ for any $z$ by \cite[equation~(C.1.1)]{dyatlov2019mathematical}, all the terms in the series must vanish. Note that $E_+(0) = R_-$ and $E_-(0) = R_+$. In particular, the term $E_-(0)E(0)^{n - 1}E_+(0)1 = 0$ implies that 
    \[ R_+E(0)^{n - 1}R_-1 = \langle \varphi_{n}^*, \varphi_1\rangle = 0. \qedhere\]
\end{proof}

Now, we consider $\alpha$ as in Theorem \ref{thm:jordan} by letting $\alpha \notin \mathcal A \cup \mathcal B$. For $u_1\in \ker_{L^2_0}D_n(\alpha)$ and $v_1\in \ker_{L^2_0}D_n(\alpha)^*$, we denote 
$\varphi = (u_1,0)^T, \psi = (0,v_1)^T$ such that $\ker_{L^2_0} H(\alpha) = \operatorname{span}_{\CC}\{\varphi, \psi\}$. We consider the Grushin problem 
\begin{equation} 
\label{eq:Grushin1}
\mathcal P_k = \begin{pmatrix}
    H_{k}(\alpha) - z & R_- \\ R_+ & 0
\end{pmatrix} \colon H^1_0(\CC; \CC^{4n}) \oplus \CC^2 \to L^2_0(\CC; \CC^{4n}) \oplus \CC^2, \end{equation} with
\[R_- \colon (c_1,c_2)^T \mapsto c_1 \varphi + c_2 \psi,\quad R_+ \colon u \mapsto (\langle u, \varphi \rangle, \langle u, \psi \rangle)^T.\]
For $k=0$ the Grushin problem $\mathcal{P}_0$ in \eqref{eq:Grushin1} is invertible for $z$ near $0$ by discreteness of the spectrum of $H(\alpha)\rvert_{L^2_0}$, with
\[ \mathcal{P}^{-1}_0 = \begin{pmatrix}
    E & E_+\\E_- & E_{-+}
\end{pmatrix},\]
where 
 \begin{equation}
 \label{eq:E}
     E=(\Pi^\perp (H_0(\alpha)-z)\Pi^\perp)^{-1}\Pi^\perp,\quad E_+=R_-,\quad E_-=R_+,\quad E_{-+}= \diag (z,z) .
 \end{equation}
Here, $\Pi=R_-R_+$ is the projection to $\ker_{L^2_0} H_0(\alpha)$ and $\Pi^{\perp}=I-\Pi$. 
By \cite[Proposition C.3]{dyatlov2019mathematical}, the operator $\mathcal P_k$ in \eqref{eq:Grushin1} is invertible for $k$ sufficiently close to $0$ with inverse
 \begin{equation}
     \label{eq:Pkinv}
     \mathcal P_k^{-1} = \begin{pmatrix} F & F_+ \\ F_- & F_{-+} \end{pmatrix} \colon \  L^2_0(\CC; \CC^{4n})\oplus \CC^2 \to H^1_0(\CC; \CC^{4n}) \oplus \CC^2.
 \end{equation}
where  
\begin{equation}
\label{eq:Grushin-asymp}
    F_{-+} = E_{-+} + \sum_{m = 1}^\infty (-1)^mE_-A(EA)^{m - 1}E_+, \quad A = \begin{pmatrix}
    & \overline{k}\\ k
\end{pmatrix}.
\end{equation} 
We compute relevant terms in the above expansion at $z = 0$. 
\begin{lemm}
\label{lem:hot}
    For $z = 0$, we have $E_-A(EA)^{m - 1}E_+ = 0$ for $ 1 \leq m \le n - 1$, and \[E_-A(EA)^{n - 1}E_+ = \begin{pmatrix}
    0 & \Bar{k}^n\langle v_n, u_1\rangle\\  k^n\langle u_n, v_1\rangle & 0
\end{pmatrix}.\]
\end{lemm}

\begin{proof}
    By Theorem \ref{thm:main}, we have for $j = 2, \dots, n,$ \begin{equation}
        H_0(\alpha)\begin{pmatrix}
            u_j \\ 0
        \end{pmatrix} = \begin{pmatrix}
            0 \\ u_{j - 1}
        \end{pmatrix},\quad H_0(\alpha) \begin{pmatrix}
            0 \\ v_j
        \end{pmatrix} = \begin{pmatrix}
            v_{j - 1} \\ 0
        \end{pmatrix}. 
    \end{equation} Since Lemma \ref{lem:dichotomy} implies that $(0, u_{j - 1})^T \perp \psi$ and $(v_{j - 1}, 0)^T\perp \varphi$, using  \eqref{eq:E} we conclude that $E: (0, u_{j - 1})^T \mapsto (u_j, 0)^T$ and $E: (v_{j - 1}, 0)^T \mapsto (0, v_j)^T$. 
    Inductively, we have 
    \begin{gather*}
        (EA)^{m - 1}E_+e_1 = (EA)^{m - 1}(u_1, 0)^T = k^{m - 1}(u_{m }, 0)^T,\\
        (EA)^{m - 1}E_+e_2 = \bar k^{m - 1}(0, v_m)^T.
    \end{gather*}
    This gives us that for $m \le n$, 
    \[E_-A(EA)^{m - 1}E_+ = \begin{pmatrix}
        k^m\langle (0, u_m), (u_1, 0)\rangle & \bar k^m\langle (v_m, 0), (u_1, 0)\rangle \\k^m\langle (0, u_m), (0, v_1)\rangle & \bar k^m\langle (v_m, 0), (0, v_1)\rangle
    \end{pmatrix}.\] 
    It is easy to see that the diagonal terms vanish, and the inner products in the off-diagonal terms vanish for $1 \le m \le n - 1$ due to Lemma \ref{lem:dichotomy}.
\end{proof}

Note that all components of $\mathcal P_k^{-1}$ are holomorphic in $z$. We conclude that 
\[  F_{-+} = \begin{pmatrix}
    z & \Bar{k}^n\langle v_n, u_1\rangle\\  k^n\langle u_n, v_1 \rangle & z
\end{pmatrix} + \begin{pmatrix}
    0 & {\mathcal{O}(|k|^{n+1})}\\ {\mathcal{O}(|k|^{n+1})} & 0
\end{pmatrix}  +  z R(k,z),\]
where $R(k,z) = \mathcal{O}(|k|)$ is smooth in $(k,\Bar{k})$-variables and holomorphic in $z$. The first term comes from Lemma \ref{lem:hot} and the second term comes from the higher order terms in $F_{-+}$ at $z=0$. Hence, $F_{-+}$ non-invertible implies  
\begin{equation}
\label{eq:det-grushin}
    z  = \pm |\langle u_{n}, v_1 \rangle| \cdot |k|^n +  {\mathcal{O}(|k|^{n+1})} + z R'(k,z), \quad R'(k,z) = \mathcal{O}(|k|).
\end{equation}
This yields that near $k=0$, zeros of $\det F_{-+}$ are given by
\begin{equation}
    z(k) = \pm |\langle u_{n}, v_1 \rangle| \cdot |k|^n + \mathcal{O}(|k|^{n+1})
\end{equation}
which gives the band structure of $H(\alpha)$ near $k=0$ by \cite[Lemma~2.10]{zworskipde}. In addition, for $\alpha \notin \mathcal A \cup \mathcal B$, we have $\langle u_n, v_1\rangle \neq 0$ by Lemma \ref{lem:dichotomy}. This proves part \eqref{part1} of Theorem \ref{thm:main}.



\section{Dirac cones and degenerated band crossing} 
\label{sec:cone}

In this section, we perform the analysis for $\alpha \in \mathcal B$ and prove part \eqref{part3} of Theorem \ref{thm:main}. Recall that by Corollary \ref{cor:dim}, we now have $\dim \ker_{L^2_0} D_{n}(\alpha) = \dim \ker_{L^2_0} D_{n}(\alpha)^* = 2$, which implies that $E_{\pm 1}(\alpha; 0) = E_{\pm 2}(\alpha; 0) = 0$ is a degenerated eigenvalue of multiplicity four.  
By the proof of Theorem \ref{thm:jordan}, there exists a set of orthonormal functions $\{u', u_1, \dots, u_{n-1}\}$ and a set of orthonormal functions $\{v', v_1, \dots, v_{n-1}\}$ such that 
\begin{gather*}
    \operatorname{span}_{\CC} \{u', u_1\} = \ker_{L^2_0} D(\alpha), \quad D(\alpha)u_i = t u_{i-1},\ i = 2, \cdots, n-1;\\
    \operatorname{span}_{\CC} \{v', v_1\} = \ker_{L^2_0} D(\alpha)^*, \quad D(\alpha)^*v_i = t v_{i-1},\ i = 2, \cdots, n-1.
\end{gather*}

We have an analogy of Lemma \ref{lem:dichotomy}. 
\begin{lemm}\label{lem:dichotomy-2}
    In the above notations, for $\alpha \in \mathcal B$, we have
    \begin{enumerate}
        \item $\langle u_i, v_j\rangle = 0$ for $i + j \leq n - 1$. Moreover, for $i < n - 1$ we have $\langle u_i, v'\rangle = \langle u', v_i\rangle = 0.$
        \item $\langle u_{n - 1}, v_1\rangle = \langle u_1, v_{n - 1}\rangle \ne 0$.
        \item $\langle u', v'\rangle \ne 0$. Moreover, $u'$ and $v'$ can be chosen such that $\langle u', v_{n - 1} \rangle = \langle u_{n - 1}, v'\rangle = 0$.
    \end{enumerate}
\end{lemm} 

\begin{proof}
    The proof of (1) and (2) is analogous to the proof of Lemma \ref{lem:dichotomy}, noting that $\alpha \in \mathcal B$ implies that $D(\alpha)-z$ is invertible in some punctured disk around $0$. To prove (3), we first choose $u'$ and $v'$ such that the inner products $\langle u', v_{n - 1}\rangle$ and $\langle u_{n - 1}, v'\rangle$ vanish. To do this, we simply redefine $u'$ and $v'$ to be 
    \[\Tilde{u}' := u' - \frac{\langle u', v_{n - 1}\rangle}{\langle u_1, v_{n - 1}\rangle}u_1, \quad \Tilde{v}' := v' - \frac{\langle u_{n - 1}, v'\rangle}{\langle u_{n - 1}, v_1\rangle}v_1, \] 
    such that the desired inner products vanish under this redefinition. As  $\Pi u' = u'$ and $\langle u', v_{i}\rangle = 0$ for $1\leq i \leq n-1$, we see that $\langle u', v'\rangle$ cannot vanish by equation \eqref{eq:proj}.
\end{proof}

Now, we consider two separate Grushin problems. Define $\varphi = (u_1, 0)^T, \varphi' = (u', 0)^T, \psi = (0, v_1)^T, \psi' = (0, v')^T \in \ker_{L^2_0} H(\alpha)$. We make the following definitions for $R_-, R_-' \colon \CC^2 \to L^2_0(\CC; \CC^{4n})$ and $R_+, R_+' \colon L^2_0(\CC, \CC^{4n}) \to \CC^2$:
\begin{gather*}
    R_- \colon (z_1, z_2)^T \mapsto z_1\varphi + z_2\psi,\quad  R_-' \colon (w_1, w_2)^T \mapsto w_1\varphi' + w_2\psi',\\
    R_+ \colon u \mapsto (\langle u, \varphi\rangle, \langle u, \psi\rangle)^T,\quad R_+' \colon u \mapsto (\langle u, \varphi'\rangle, \langle u, \psi'\rangle)^T.
\end{gather*}
Now, consider the iterative Grushin problems
\begin{equation*}
    \mathcal P_k = \begin{pmatrix}
        \begin{bmatrix}
            H_k(\alpha) - z & R_-' \\ 
            R_+'
        \end{bmatrix} & \begin{bmatrix}
            R_- \\ 0
        \end{bmatrix} \\ \begin{bmatrix}
           \ \quad R_+ \ \quad  & \ \  0 \ \  
        \end{bmatrix}
    \end{pmatrix} \colon (L^2_0(\CC; \CC^{4n}) \oplus \CC^2) \oplus \CC^2 \to (L^2_0(\CC; \CC^{4n}) \oplus \CC^2) \oplus \CC^2
\end{equation*}
and
\begin{equation*}
    \mathcal P_k' = \begin{pmatrix}
        \begin{bmatrix}
            H_k(\alpha) - z & R_- \\ 
            R_+
        \end{bmatrix} & \begin{bmatrix}
            R_-' \\ 0
        \end{bmatrix} \\ \begin{bmatrix}
             \ \quad R'_+ \ \quad  & \ \  0 \ \ 
        \end{bmatrix}
    \end{pmatrix} \colon (L^2_0(\CC; \CC^{4n}) \oplus \CC^2) \oplus \CC^2 \to (L^2_0(\CC; \CC^{4n}) \oplus \CC^2) \oplus \CC^2.
\end{equation*}

Now, for $k = 0$, both $\mathcal P_0$ and $\mathcal P_0'$ are invertible, and we may denote the corresponding inverses by \begin{equation}
    \mathcal P_0^{-1} = \begin{pmatrix}
        E & E_+ \\ E_- & E_{-+}
    \end{pmatrix}, \quad (\mathcal P_0')^{-1} = \begin{pmatrix}
        E' & E_+' \\ E_-' & E_{-+}'
    \end{pmatrix}
\end{equation}
Similarly to the previous section, it is easy to check that 
\begin{gather*}
    E_-= R_+, \quad E_+ = R_-, \quad E_-' = R_+' , \quad E_+' = R_-', \quad E_{-+} = E_{-+}' = \diag(z, z),\\
    E = \left(\Pi^\bot \begin{bmatrix}
        H_k(\alpha) - z & R_-' \\ R_+'
    \end{bmatrix}\Pi^\bot\right)^{-1}\Pi^\bot, \quad E' = \left((\Pi')^\bot \begin{bmatrix}
        H_k(\alpha) - z & R_- \\ R_+
    \end{bmatrix}(\Pi')^\bot\right)^{-1}(\Pi')^\bot,
\end{gather*}
where $\Pi = R_-R_+$ is the projection to $\operatorname{span}_{\CC}\{\varphi, \psi\}$, $\Pi' = R_-'R_+'$ is the projection to $\operatorname{span}_{\CC}\{\varphi', \psi'\}$, and $\Pi^\bot = I - \Pi,\ (\Pi')^\bot = I - \Pi'$. Note that $E, E'$ are well-defined this way by Lemma \ref{lem:dichotomy-2}. 

For $k$ sufficiently small, both $\mathcal P_k$ and $\mathcal P_k'$ are invertible, with the inverses \[\mathcal P_k^{-1} = \begin{pmatrix}
    F & F_+\\ F_- & F_{-+} 
\end{pmatrix}, \quad (\mathcal P_k')^{-1} = \begin{pmatrix}
    F & F_+' \\ F_-' & F_{-+'}
\end{pmatrix},\] where we have that $F_+ = E_+, F_- = E_-, F_+' = E_+', F_-' = E_-'$ by \cite[Appendix C.3]{dyatlov2019mathematical}. We also obtain the series expansions 
\begin{gather*}
     F_{-+} = E_{-+} + \sum_{m = 1}^\infty (-1)^mE_-A(EA)^{m - 1}E_+,\quad 
      F_{-+}' = E_{-+}' + \sum_{m = 1}^\infty (-1)^mE_-'A(E'A)^{m - 1}E_+',
\end{gather*}
with $A$ in \eqref{eq:Grushin-asymp}. As Lemma \ref{lem:dichotomy-2} indicates that that $(0, u_i)^T ,(v_i, 0)^T \in \ker \Pi'$ for $1 \le i \le n - 1$, we may seamlessly apply the same computation as in Lemma \ref{lem:hot} to conclude that \begin{equation}
    F_{-+} = \begin{pmatrix}
        z & \bar k^{n - 1} \langle v_{n -1}, u_1\rangle \\ k^{n - 1}\langle u_{n - 1}, v_1\rangle & z
    \end{pmatrix} + \begin{pmatrix}
        0 & \mathcal O(|k|^n) \\ \mathcal O(|k|^n) & 0
    \end{pmatrix} + zR(k, z),
\end{equation} where $R = \mathcal O(|k|)$. This indicates that $F_{-+}$ is singular when $z = \pm|\langle u_{n - 1}, v_1\rangle||k|^{n - 1} + \mathcal O(|k|^{n})$ for sufficiently small $k$. Similarly, because $(0, u')^T, (v', 0)^T \in \ker \Pi$, we have \[F_{-+}' = \begin{pmatrix}
    z & \bar k\langle v', u'\rangle \\ k\langle u', v'\rangle & z
\end{pmatrix} + \begin{pmatrix} 0
    & \mathcal O(|k|^2)\\ \mathcal O(|k|^2) & 0
\end{pmatrix},\] which implies that $F_{-+}'$ is singular when $z = \pm|\langle u', v'\rangle ||k| + \mathcal O(|k|^2)$. 
Note that $\det F_{-+} = 0$ implies that \[\mathcal H_k := \begin{bmatrix}
    H_k(\alpha) - z & R_-' \\ R_+'
\end{bmatrix}\colon L^2_0(\CC; \CC^{4n}) \oplus \CC^2 \to L_0^2(\CC; \CC^{4n}) \oplus \CC^2\] is non-invertible. Hence $H_k(\alpha) - z$ is not invertible, as if $H_k(\alpha) - z$ were invertible, then so would $\mathcal H_k$. The same argument holds for $F_{-+}'$. We obtain the eigenvalues of $H_k(\alpha)$ are 
\[E_{\pm 1}(k) = \pm|\langle u_{n - 1}, v_1\rangle||k|^{n - 1} + \mathcal O(|k|^{n}), \quad E_{\pm 2}(k) = \pm|\langle u', v'\rangle ||k| + \mathcal O(|k|^2)\] 
for $k$ sufficiently small. Lemma \ref{lem:dichotomy-2} implies the leading order coefficients are nonzero.

\bibliography{reference}

\newcommand{\etalchar}[1]{$^{#1}$}
\begin{thebibliography}{WKML23}

\bibitem[BC18]{berkolaiko2018symmetry}
Gregory Berkolaiko and Andrew Comech.
\newblock Symmetry and {D}irac points in graphene spectrum.
\newblock {\em Journal of Spectral Theory}, 8(3):1099--1147, 2018.

\bibitem[BEWZ21]{becker2021spectral}
Simon Becker, Mark Embree, Jens Wittsten, and Maciej Zworski.
\newblock Spectral characterization of magic angles in twisted bilayer
  graphene.
\newblock {\em Physical Review B}, 103(16):165113, 2021.

\bibitem[BEWZ22]{becker2020mathematics}
Simon Becker, Mark Embree, Jens Wittsten, and Maciej Zworski.
\newblock Mathematics of magic angles in a model of twisted bilayer graphene.
\newblock {\em Probability and Mathematical Physics}, 3(1):69--103, 2022.

\bibitem[BHWY25]{becker2025chiral}
Simon Becker, Tristan Humbert, Jens Wittsten, and Mengxuan Yang.
\newblock Chiral limit of twisted trilayer graphene.
\newblock {\em Nonlinearity}, 38(5):055008, 2025.

\bibitem[BHZ23a]{becker2023degenerate}
Simon Becker, Tristan Humbert, and Maciej Zworski.
\newblock Degenerate flat bands in twisted bilayer graphene.
\newblock {\em arXiv preprint arXiv:2306.02909}, 2023.

\bibitem[BHZ23b]{becker2022integrability}
Simon Becker, Tristan Humbert, and Maciej Zworski.
\newblock Integrability in the chiral model of magic angles.
\newblock {\em Communications in Mathematical Physics}, 403(2):1153--1169,
  2023.

\bibitem[BHZ24]{becker2022fine}
Simon Becker, Tristan Humbert, and Maciej Zworski.
\newblock Fine structure of flat bands in a chiral model of magic angles.
\newblock In {\em Annales Henri Poincar{\'e}}, pages 1--31. Springer, 2024.

\bibitem[BKZ24]{zhu2024magnetic}
Simon Becker, Jihoi Kim, and Xiaowen Zhu.
\newblock Magnetic response properties of twisted bilayer graphene.
\newblock In {\em Annales Henri Poincar{\'e}}, pages 1--46. Springer, 2024.

\bibitem[BLS25]{becker2025exact}
Simon Becker, Lin Lin, and Kevin Stubbs.
\newblock Exact ground state of interacting electrons in magic angle graphene.
\newblock {\em Communications in Mathematical Physics}, 406(6):148, 2025.

\bibitem[BM11]{bistritzer2011moire}
Rafi Bistritzer and Allan MacDonald.
\newblock Moir{\'e} bands in twisted double-layer graphene.
\newblock {\em Proceedings of the National Academy of Sciences},
  108(30):12233--12237, 2011.

\bibitem[BOV23]{becker2023magic}
Simon Becker, Izak Oltman, and Martin Vogel.
\newblock Magic angle (in) stability and mobility edges in disordered chern
  insulators.
\newblock {\em arXiv preprint arXiv:2309.02701}, 2023.

\bibitem[BOV24]{becker2024absence}
Simon Becker, Izak Oltman, and Martin Vogel.
\newblock Absence of small magic angles for disordered tunneling potentials in
  twisted bilayer graphene.
\newblock {\em arXiv preprint arXiv:2402.12799}, 2024.

\bibitem[BQT{\etalchar{+}}24]{becker2024dirac}
Simon Becker, Solomon Quinn, Zhongkai Tao, Alexander Watson, and Mengxuan Yang.
\newblock Dirac cones and magic angles in the {B}istritzer--{M}acdonald {TBG}
  {H}amiltonian.
\newblock {\em arXiv preprint arXiv:2407.06316}, 2024.

\bibitem[BZ24a]{becker2023dirac}
Simon Becker and Maciej Zworski.
\newblock Dirac points for twisted bilayer graphene with in-plane magnetic
  field.
\newblock {\em Journal of Spectral Theory}, 14(2):479--511, 2024.

\bibitem[BZ24b]{becker2024chiral}
Simon Becker and Maciej Zworski.
\newblock From the chiral model of {TBG} to the {B}istritzer--{M}acdonald
  model.
\newblock {\em Journal of Mathematical Physics}, 65(6), 2024.

\bibitem[BZ25]{zhu2025spectral}
Simon Becker and Xiaowen Zhu.
\newblock Spectral theory of twisted bilayer graphene in a magnetic field.
\newblock {\em SIAM Journal on Mathematical Analysis}, 57(2):1621--1651, 2025.

\bibitem[CFD{\etalchar{+}}18]{cao2018correlated}
Yuan Cao, Valla Fatemi, Ahmet Demir, Shiang Fang, Spencer Tomarken, Jason Luo,
  Javier Sanchez-Yamagishi, Kenji Watanabe, Takashi Taniguchi, Efthimios
  Kaxiras, et~al.
\newblock Correlated insulator behaviour at half-filling in magic-angle
  graphene superlattices.
\newblock {\em Nature}, 556(7699):80--84, 2018.

\bibitem[CFF{\etalchar{+}}18]{cao2018unconventional}
Yuan Cao, Valla Fatemi, Shiang Fang, Kenji Watanabe, Takashi Taniguchi,
  Efthimios Kaxiras, and Pablo Jarillo-Herrero.
\newblock Unconventional superconductivity in magic-angle graphene
  superlattices.
\newblock {\em Nature}, 556(7699):43--50, 2018.

\bibitem[CGG23]{PhysRevB.107.155403}
\'Eric Canc\`es, Louis Garrigue, and David Gontier.
\newblock Simple derivation of moir\'e-scale continuous models for twisted
  bilayer graphene.
\newblock {\em Phys. Rev. B}, 107:155403, Apr 2023.

\bibitem[CM23]{cances2023semiclassical}
\'Eric Canc{\`e}s and Long Meng.
\newblock Semiclassical analysis of two-scale electronic {H}amiltonians for
  twisted bilayer graphene.
\newblock {\em arXiv preprint arXiv:2311.14011}, 2023.

\bibitem[DZ19]{dyatlov2019mathematical}
Semyon Dyatlov and Maciej Zworski.
\newblock {\em Mathematical theory of scattering resonances}, volume 200.
\newblock American Mathematical Soc., 2019.

\bibitem[DZZ25a]{zworski2025scalar}
Semyon Dyatlov, Henry Zeng, and Maciej Zworski.
\newblock Private communications.
\newblock 2025.

\bibitem[DZZ25b]{dyatlov2025scalar}
Semyon Dyatlov, Henry Zeng, and Maciej Zworski.
\newblock {WKB} structure in a scalar model of flat bands.
\newblock {\em preprint}, 2025.

\bibitem[FW12]{fefferman2012honeycomb}
Charles Fefferman and Michael Weinstein.
\newblock Honeycomb lattice potentials and {D}irac points.
\newblock {\em Journal of the American Mathematical Society}, 25(4):1169--1220,
  2012.

\bibitem[GZ23]{galkowski2023abstract}
Jeffrey Galkowski and Maciej Zworski.
\newblock An abstract formulation of the flat band condition.
\newblock {\em arXiv preprint arXiv:2307.04896}, 2023.

\bibitem[HZ25]{hitrik2025classically}
Michael Hitrik and Maciej Zworski.
\newblock Classically forbidden regions in the chiral model of twisted bilayer
  graphene.
\newblock {\em Probability and Mathematical Physics}, 6(2):505--546, 2025.

\bibitem[LTKV20]{ledwith2020fractional}
Patrick Ledwith, Grigory Tarnopolsky, Eslam Khalaf, and Ashvin Vishwanath.
\newblock Fractional chern insulator states in twisted bilayer graphene: An
  analytical approach.
\newblock {\em Physical Review Research}, 2(2):023237, 2020.

\bibitem[LVK22]{ledwith2022family}
Patrick Ledwith, Ashvin Vishwanath, and Eslam Khalaf.
\newblock Family of ideal {C}hern flatbands with arbitrary {C}hern number in
  {C}hiral twisted graphene multilayers.
\newblock {\em Physical Review Letters}, 128(17):176404, 2022.

\bibitem[PCW{\etalchar{+}}21]{park2021tunable}
Jeong~Min Park, Yuan Cao, Kenji Watanabe, Takashi Taniguchi, and Pablo
  Jarillo-Herrero.
\newblock Tunable strongly coupled superconductivity in magic-angle twisted
  trilayer graphene.
\newblock {\em Nature}, 590(7845):249--255, 2021.

\bibitem[QKLW25]{quinn2025higher}
Solomon Quinn, Tianyu Kong, Mitchell Luskin, and Alexander Watson.
\newblock Higher-order continuum models for twisted bilayer graphene.
\newblock {\em arXiv preprint arXiv:2502.08120}, 2025.

\bibitem[RB11]{PhysRevX.1.021014}
Nicolas. Regnault and Andrei Bernevig.
\newblock Fractional {C}hern insulator.
\newblock {\em Phys. Rev. X}, 1:021014, Dec 2011.

\bibitem[SBL24]{stubbs2024hartree}
Kevin Stubbs, Simon Becker, and Lin Lin.
\newblock On the {H}artree-{F}ock ground state manifold in magic angle twisted
  graphene systems.
\newblock {\em arXiv preprint arXiv:2403.19890}, 2024.

\bibitem[SRML25]{stubbs2025many}
Kevin Stubbs, Michael Ragone, Allan MacDonald, and Lin Lin.
\newblock The many-body ground state manifold of flat band interacting
  {H}amiltonian for magic angle twisted bilayer graphene.
\newblock {\em arXiv preprint arXiv:2503.20060}, 2025.

\bibitem[TKV19]{tarnopolsky2019origin}
Grigory Tarnopolsky, Alex Kruchkov, and Ashvin Vishwanath.
\newblock Origin of magic angles in twisted bilayer graphene.
\newblock {\em Physical Review Letters}, 122(10):106405, 2019.

\bibitem[TZ23]{zworskipde}
Zhongkai Tao and Maciej Zworski.
\newblock {PDE} methods in condensed matter physics.
\newblock {\em Lecture Notes}, 2023.

\bibitem[WKML23]{watson2023bistritzer}
Alexander Watson, Tianyu Kong, Allan MacDonald, and Mitchell Luskin.
\newblock Bistritzer--{M}acdonald dynamics in twisted bilayer graphene.
\newblock {\em Journal of Mathematical Physics}, 64(3), 2023.

\bibitem[WL21]{watson2021existence}
Alexander Watson and Mitchell Luskin.
\newblock Existence of the first magic angle for the chiral model of bilayer
  graphene.
\newblock {\em Journal of Mathematical Physics}, 62(9):091502, 2021.

\bibitem[WL22]{wang2022hierarchy}
Jie Wang and Zhao Liu.
\newblock Hierarchy of ideal flatbands in chiral twisted multilayer graphene
  models.
\newblock {\em Physical Review Letters}, 128(17):176403, 2022.

\bibitem[XPP{\etalchar{+}}21]{xie2021fractional}
Yonglong Xie, Andrew Pierce, Jeong~Min Park, Daniel Parker, Eslam Khalaf,
  Patrick Ledwith, Yuan Cao, Seung~Hwan Lee, Shaowen Chen, Patrick Forrester,
  et~al.
\newblock Fractional {C}hern insulators in magic-angle twisted bilayer
  graphene.
\newblock {\em Nature}, 600(7889):439--443, 2021.

\bibitem[Yan23]{yang2023flat}
Mengxuan Yang.
\newblock Flat bands and high {C}hern numbers in twisted multilayer graphene.
\newblock {\em Journal of Mathematical Physics}, 64(11), 2023.

\bibitem[Zwo24]{zworski2023survey}
Maciej Zworski.
\newblock Mathematical results on the chiral models of twisted bilayer graphene
  (with an appendix by {M}engxuan {Y}ang and {Z}hongkai {T}ao).
\newblock {\em Journal of Spectral Theory}, 14(3):1063--1107, 2024.

\end{thebibliography}
\bibliographystyle{alpha}

\end{document}